\documentclass{article}
\RequirePackage{fix-cm}
\usepackage{setspace}

\usepackage[misc,geometry]{ifsym} 

\usepackage{color}
\usepackage{srcltx}
\usepackage{fancyhdr}
\usepackage{latexsym}
\usepackage{amsfonts}
\usepackage{amscd}
\usepackage{epsfig}
\usepackage[cp850]{inputenc}
\usepackage[english]{babel}
\usepackage{ifthen}
\usepackage{float}
\usepackage{amsmath}
\usepackage{amssymb}
\usepackage{mathrsfs}
\usepackage[mathscr]{eucal}
\usepackage{dsfont}
\usepackage{amstext}
\usepackage{syntonly}

\usepackage{tikz} 
\usepackage{color}
\usepackage{graphicx}

\usepackage{enumerate}
\usepackage{amssymb}
\usepackage[english]{babel}
\usepackage{amsthm}
\usepackage{amstext}
\usepackage{graphicx}
 \usepackage{mathptmx}      
\usepackage{latexsym}
  \usepackage{enumerate}

\newtheorem{theorem}{Theorem}[section] 
\newtheorem{lemma}[theorem]{Lemma}     
\newtheorem{corollary}[theorem]{Corollary}
\newtheorem{proposition}[theorem]{Proposition}
\newtheorem{definition}[theorem]{Definition}   

\newtheorem{example}[theorem]{Example}

\usepackage{color}

\usepackage{enumerate}
\usepackage{amssymb}
\usepackage[english]{babel}
\usepackage{amstext}
\usepackage{graphicx}
 \usepackage{mathptmx}      
\usepackage{latexsym}
\newcommand{\footremember}[2]{%
    \footnote{#2}
    \newcounter{#1}
    \setcounter{#1}{\value{footnote}}%
}

\date{}
\author{%
 A. Estevan\footremember{alley}{Dpto. Estad\'istica, Inform\'atica y Matem\'aticas, Institute INAMAT, Public University of Navarre.   Campus Arrosad\'{\i}a, 31006.  Iru\~na-Pamplona, Navarre, Spain.    }%
  \and J.J. Mi\~nana\footremember{trailer}{Departament de Ci\`encies Matem\`atiques i Inform\`atica,
Universitat de les Illes Balears, Ctra. de Valldemossa km. 7.5,
07122 Palma de Mallorca, Spain.}  %
  \and O. Valero\footremember{alleydd}{Departament de Ci\`encies Matem\`atiques i Inform\`atica,
Universitat de les Illes Balears, Ctra. de Valldemossa km. 7.5,
07122 Palma de Mallorca, Spain. } \footnote{Corresponding author. }%
  }

\title{On fixed point theory in partially ordered sets and an application to asymptotic complexity of algorithm}

\begin{document}
\maketitle
\begin{abstract}
The celebrated Kleene fixed point theorem is crucial in the mathematical modelling of recursive specifications in Denotational Semantics. In this paper we discuss whether the hypothesis of the aforementioned result can be weakened. An affirmative answer to the aforesaid inquiry is provided so that a characterization of those properties that a self-mapping must satisfy in order to guarantee that its set of fixed points is non-empty when no notion of completeness are assumed to be satisfied by the partially ordered set. Moreover, the case in which the partially ordered set is coming from a quasi-metric space is treated in depth. Finally, an application of the exposed theory is obtained. Concretely, a mathematical method to discuss the asymptotic complexity of those algorithms whose running time of computing  {fulfills} a recurrence equation is presented.  {Moreover, the aforesaid method retrieves the fixed point based methods that appear in the literature for asymptotic complexity analysis of algorithms. However, our new method improves the aforesaid methods because it imposes fewer requirements than those that have been assumed in the literature and, in addition, it allows to state simultaneously upper and lower asymptotic bounds for the running time computing.}

\end{abstract}

keywords: partial order, quasi-metric, fixed point, Kleene, asymptotic complexity, recurrence equation.

\section{Introduction}\label{Int}
Fixed point theory in partially ordered sets plays a central role in  {the research activity in Mathematics and Computer Science (\cite{Davey,JGL2013,SedaHit2,RusPetrusel2008,Stoy})}. In particular, Kleene's fixed point theorem is one of the fundamental pillars of Denotational Semantics (see, for instance, \cite{Davey,Manna,Stoy}). The aforesaid result allows to state the so-called Scott's induction principle which models the meaning of recursive specifications in programming languages as the fixed point of non-recursive monotone self-mappings defined in partially ordered sets, in such a way that the aforesaid fixed point is the supremum of the sequence of successive iterations of the non-recursive mapping acting on a distinguished element of the model (see \cite{Scott2,Scott}). In Scott's approach, the non-recursive mapping models the evolution of the program execution and the partial order encodes some computational information notion so that each iteration of the mapping matches up with an element of the mathematical model which is greater than (or equal to) those that are associated to the preceding steps of the computational process. It is assumed that in each step the computational process gives more information about the meaning of the denotational specification than the preceding steps. Therefore, the aforementioned fixed point encodes the total information about the meaning provided by the elements of the increasing sequence of successive iterations and, in addition, no more information can be extracted by the fixed point than that provided by each element of such a sequence.

In order to guarantee the existence of fixed point of a monotone self-mapping, Kleene's fixed point theorem assumes conditions about the partially ordered set (order-completeness) and the self-mapping (order-continuity).  {However, in many real applications one of two conditions can be unfulfilled. Motivated, in part, by this fact a few works have focused their efforts on generalized versions of Kleene's fixed point theorem recently (see, for instance, \cite{EsikRon,FomenkoPodop16,FomenkoPodop17}). In the original version of the celebrated Kleene fixed point theorem, and also in the aforesaid references, the assumed conditions} have a global character, i.e., each element of the partially ordered set (the mathematical model) must  {satisfy} them. However, in the  {aforementioned real applications, coming, for example, from Denotational Semantics or Logic Programming}, to check the aforesaid conditions for all elements of the partially ordered set is unnecessary. In fact, the proof of Kleene's fixed point theorem is based on the construction of a sequence of iterations from a fixed element and, thus, the global assumed conditions apply for warranting the desired conclusions. In the view of the preceding remark, it seems natural to wonder whether the hypothesis in the statement of Kleene's fixed point theorem can be weakened in such a way that the new ones are better suited to the demands of the real problems (with local more than global character) and, at the same time, preserve the spirit of the original Kleene's fixed point theorem.

In this paper we provide an affirmative answer to the  {question posed}. Concretely, we characterize those properties that a self-mapping must  {satisfy} in order to ensure that its set of fixed points is non-empty when a general partially ordered set is under consideration and no notion of order-completeness is assumed. Moreover, we derive a few characterization when, in addition, the partially ordered set is chain complete and the self-mapping is order-continuous. Special interest is paid to that case in which the partially ordered set is coming from a quasi-metric space,  {since such generalized distances have shown to be useful in Denotational Semantics, Logic Programming and Asymptotic Complexity of algorithms (\cite{JGL2013,SedaHit2,Schellekens})}.  {Finally, the developed theory is applied to discuss the asymptotic complexity of those algorithms whose running time of computing fulfills a recurrence equation. Thus, on the one hand, a fixed point method for asymptotic complexity is developed in such a way that those fixed point methods given in \cite{RPV,RomVal2013,Schellekens} and that are based on the use of contractive mappings are retrieved as a particular case. On the other hand, the aforementioned new fixed point method captures the essence of that for discussing the asymptotic complexity of Probabilistic Divide and Conquer algorithms given in \cite{GRRS2008}. Nonetheless, our new method improves the aforesaid methods because it imposes fewer requirements than those that have been assumed in the literature and, in addition, it allows to state simultaneously upper and lower asymptotic bounds for the running time computing. Besides, the new fixed point method also preserves the original Scott's ideas providing a common framework for Denotational Semantics and Asymptotic Complexity of algorithms.}

\section{The fixed point theorems}\label{Results}
This section is devoted to discern which are the minimal conditions that allow to guarantee the existence of fixed point for a self-mapping defined in partially ordered sets. In order to achieve our objective we recall a few pertinent notions.

Following \cite{Davey}, a partially ordered set is a pair $(X,\preceq)$ such that $X$ is a nonempty set and $\preceq$ is a
binary relation on $X$ which holds, for all $x,y,z\in X$:

$
\begin{array}{l}
\mathrm{(i)}\text{ }x\preceq x \\
\mathrm{(ii)}\text{ }x\preceq y\text{\textrm{\ and }}y\preceq x \Rightarrow x=y \\
\mathrm{(iii)}\text{ }x\preceq y\text{\textrm{\ and }}y\preceq z \Rightarrow x\preceq z%
\end{array}%
\begin{array}{l}
\text{\textrm{(reflexivity)}}, \\
\text{\textrm{(antisymmetry),}} \\
\text{\textrm{(transitivity).}}%
\end{array} $

If $(X,\preceq)$ is a partially ordered set and $Y\subseteq X$, then an upper bound for $Y$ in $(X,\preceq)$ is an element $x\in X$ such that $y\preceq x$ for all $y\in Y$. An element $z\in Y$ is the minimum of $Y$ in $(X,\preceq)$ provided that $z\preceq y$ for all $y\in Y$. Thus, the supremum of  $Y$ in $(X,\preceq)$, if exists, is an element $x^\star\in X$ which is an upper bound for $Y$ and, in addition, it is the minimum of the set $(UB(Y),\preceq)$, where $UB(Y)=\{u\in X:u \mbox{ is an upper bound for }Y \}$. Moreover, fixed $x\in X$, the sets $\{y\in X: x\preceq y\}$ and $\{y\in X: y\preceq x\}$ will be denoted by $\uparrow_{\preceq} x$ and $\downarrow_{\preceq} x$, respectively.

According to \cite{Baranga}, a partially ordered set $(X,\preceq)$ is said to be chain complete provided that there exists the supremum of every increasing sequence. Of course, a sequence $(x_{n})_{n\in \mathbb{N}^\star}$ is said to be increasing whenever $x_{n}\preceq x_{n+1}$ for all $n\in \mathbb{N}$, where $\mathbb{N}^\star$ denotes the set $\mathbb{N}\cup\{0\}$ and $\mathbb{N}$ denotes the set of positive integer numbers. 

After recalling the above notions on partially ordered sets, we present the well-known Kleene's fixed point theorem (see \cite{Baranga,Davey,Manna,Stoy}). First, let us recall that a mapping $f:X\rightarrow X$ is said to be $\preceq$-continuous provided that the supremum of the sequence $(f(x_{n}))_{n\in\mathbb{N}^\star}$ is $f(x)$ for every increasing sequence $(x_{n})_{n\in \mathbb{N}}$ whose supremum in $(X,\preceq)$ exists and is $x$.

\begin{theorem}\label{monotdcpo2}Let $(X,\preceq)$ be a chain complete partially ordered set and let $f:X\rightarrow X$ be a $\preceq$-continuous mapping. Assume that there exist $x_{0}\in X$ such that $x_{0}\preceq f(x_{0})$.

Then, there exist a fixed point  $x^{\star}$ which is supremum of the sequence $(f^n(x_0))_{n\in\mathbb{N}^\star}$ and, thus, $x^{\star}\in \uparrow_{\preceq}x_{0}$. Moreover, $x^{\star}\in \downarrow_{\preceq}y_{0}$ provided that $y_{0}\in X$ such that $x_{0}\preceq y_{0}$ and $f(y_{0})\preceq y_{0}$. Furthermore, $x^{\star}$ is the minimum fixed point in $\uparrow_{\preceq}x_{0}$.
\end{theorem}

It is well known that each $\preceq$-continuous mapping is monotone. So, Kleene's theorem cannot be applied, at least, to non-monotone mappings. However, the next example shows that there are self-mappings on a chain complete partially ordered set, which fulfill the conclusions of the above theorem, but there are not monotone (and consequently, there are not $\preceq$-continuous).

\begin{example}\label{non_order_cont}Consider the chain complete partially ordered set $(\left[0,1\right],\leq)$,  where $\leq$ stands for the usual partial order defined on $[0,1]$. Define $f:\left[0,1\right]\rightarrow\left[0,1\right]$ by $$f(x)=\left\{\begin{array}{ll}
1-\frac{x}{2},& \text{ if } x\in\left[0,\frac{1}{2}\right[\\\\
\frac{1+x}{2}, & \text{ if } x\in \left[\frac{1}{2},1\right]
\end{array}\right..$$

On the one hand, we can observe that $f$ is not monotone on $(\left[0,1\right],\leq)$, so it is not $\leq$-continuous. Nevertheless, $f$ has as a fixed point $x=1$.

On the other hand,  $\frac{1}{2}\leq f(\frac{1}{2})$, since $f(\frac{1}{2})=\frac{3}{4}$. Furthermore, a straightforward computation shows that the sequence $(f^n(\frac{1}{2}))_{n\in\mathbb{N}^\star}$ is increasing in $([0,1],\leq)$ and, in addition, $1$ is its supremum. The rest of conclusions of Theorem \ref{monotdcpo2} are clearly obtained due to the fact that $1$ is the supremum of $\left[0,1\right]$. 
\end{example}

The preceding example suggests the possibility of providing a more general version of Kleene's fixed point theorem where weakener conditions are assumed. To this end, we introduce the following concept related to $\preceq$-continuity.

\begin{definition}
Let  $(X,\preceq)$ be a partially ordered set and let $x_{0}\in X$. A mapping $f:X\rightarrow X$ will be said to be orbitally $\preceq$-continuous at $x_{0}$ provided that $f$ preserves the supremum of the sequence $(f^{n}(x_{0}))_{n\in\mathbb{N}^\star}$, i.e., $f(x)$ is the supremum of the sequence $(f^{n+1}(x_{0}))_{n\in\mathbb{N}^\star}$ in $(X,\preceq)$, whenever $x$ is the supremum of sequence $(f^{n}(x_{0}))_{n\in\mathbb{N}^\star}$.
\end{definition}


It is not hard to check that the self-mapping defined in Example \ref{non_order_cont} is orbitally $\preceq$-continuous at $\frac{1}{2}$.

Notice that, initially, there is not a direct relationship between the preceding notion and the $\preceq$-continuity. Clearly there are $\preceq$-continuous self-mappings that are not orbitally $\preceq$-continuous such as the next example illustrates.

\begin{example}\label{NOOC}Consider the partially ordered set $(\left[0,1\right],\preceq_{1})$ where $\preceq_{1}$ is defined for all $x,y\in [0,1]$ as follows:
$$x\preceq_{1}y \Leftrightarrow x=y \text{ or } y=1.$$ Define $f:\left[0,1\right]\rightarrow\left[0,1\right]$ by $f(x)=\frac{x}{2}$. Clearly $f$ is $\preceq_{1}$-continuous, since a  sequence $(x_{n})_{n\in\mathbb{N}^{\star}}$ is increasing in $([0,1],\preceq_{1})$ provided that $x_{n}=x_{n+1}$ for all $n\in\mathbb{N}^{\star}$. However, $f$ is not orbitally $\preceq_{1}$-continuous, for instance, at $1$. Indeed, the sequence $(f^{n}(1))_{n\in\mathbb{N}^{\star}}$ is given by $$f^{n}(1)=\left\{\begin{array}{ll}
1 & \text{ if } n=0 \\
\frac{1}{2^{n}} & \text{ if } n\geq 1 \\
\end{array}\right.$$ and, thus, it has $1$ as the supremum in $([0,1],\preceq_{1})$. Nevertheless, 
the sequence $(f^{n+1}(1))_{n\in\mathbb{N}^{\star}}$ is given by 
$f^{n+1}(1)=\frac{1}{2^{n+1}} $ for all $n\in \mathbb{N}^{\star}$ and it has not $f(1)$ as the supremum in $([0,1],\preceq_{1})$.
\end{example}

It must be pointed out that, given a partially ordered set $(X,\preceq)$ and $x_{0}\in X$, every $\preceq$-continuous self-mapping is orbitally $\preceq$-continuous at $x_{0}$ whenever $x_{0}\preceq f(x_{0})$.

The next example shows that there are orbitally $\preceq$-continuous self-mappings that are not $\preceq$-continuous.

\begin{example}\label{NM2}Consider the partially ordered set $(X,\preceq_{X})$ such that $X=[0,1]\cup \{2\}$ and the partial order $\preceq_{X}$ defined on $X$ as follows:

$$x\preceq_{X}y\Leftrightarrow \left\{\begin{array}{cl}
x,y\in [0,1]   \mbox{ and } y\leq x&\\
\mbox{ or } &\\
x\in ]0,1]\cup\{2\} \mbox{ and } y=2& \\
\end{array}\right..
$$ Define the mapping $f(x)=0$ for all $x\in [0,1]$ and $f(2)=2$. It is clear that $f$ is not monotone, since $1\preceq_{X} 2$ but $0=f(1)\not\preceq_{X}f(2)=2$. So, it is not $\preceq_{X}$-continuous. It is clear that $f$ is orbitally $\preceq_{X}$-continuous at $x_{0}$, with $x_{0}\in[0,1]\cup\{2\}$. 
\end{example}

Even more, orbitally $\preceq$-continuity at any $x_{0}$ does not imply that the sequence $(f^{n}(x_{0}))_{n\in\mathbb{N}^\star}$ is increasing, as demonstrates the following example.

\begin{example}\label{ExDe}Consider the chain complete partially ordered set $(\left[0,1\right],\leq)$ introduced in Example \ref{non_order_cont}. Define $f:\left[0,1\right]\rightarrow\left[0,1\right]$ by $f(x)=\frac{x}{2}$. Take $x_0\in[0,1]$. Then, the sequence  $(f^{n}(x_0))_{n\in\mathbb{N}^\star}$ is decreasing, since $f^n(x_0)=\frac{x_0}{2^n}$, for each $n\in\mathbb{N}$. Furthermore, $x_0$ is the supremum of $(f^{n}(x_0))_{n\in\mathbb{N}^\star}$ and $\frac{x_0}{2}$ is the supremum of the sequence $(f^{n+1}(x_0))_{n\in\mathbb{N}^\star}$. Since $f(x_0)=\frac{x_0}{2}$ we have that $f$ is orbitally $\leq$-continuous at $1$.
\end{example}

Another restrictive condition of Theorem \ref{monotdcpo2} is the assumption of chain completeness of the partially ordered set. Indeed, the example below shows an instance of self-mapping defined in a non chain complete partially ordered which has a fixed point satisfying all the conclusions in the aforesaid theorem.

\begin{example}
Consider the partially ordered set $([0,2[,\leq)$, where $\leq$ stands for the usual partial order defined on $[0,2[$. Obviously, $([0,2[,\leq)$ is not chain complete. The mapping $f:[0,2[\rightarrow[0,2[$ given by $f(x)=\frac{x+1}{2}$ has $1$ as a fixed point. Moreover, the sequence $(f^{n}(0))_{n\in\mathbb{N}^\star}$ is increasing and $f$ is orbitally $\preceq$-continuous at $0$. Obviously $1$ is the supremum of $(f^{n}(0))_{n\in\mathbb{N}^\star}$ and $1\in \downarrow_{\leq}y$ such that $y\in[1,2[$ (notice that $f(y)\leq y \Leftrightarrow 1\leq y$ and $0\leq y$ for all $y\in [1,2[$).
\end{example}

In order to yield a generalized Kleene's fixed point theorem, the above exposed facts suggest the possibility of  demanding only conditions on the sequence $(f^{n}(x_{0}))_{n\in\mathbb{N}^\star}$, for a given $x_0$, in order to weaken to the maximum the assumptions in the statement of Kleene's fixed point theorem.

The next result shows that such a Kleene type fixed point is possible in the suggested direction in such a way that it provides two characterizations of those properties that a self-mappings must satisfy in order to have a fixed point in partially ordered sets (without order-completeness assumptions). Before stating it, let us point out that, given a partially ordered set $(X,\preceq)$ and a mapping $f:X\rightarrow X$, we will denote by $Fix(f)$ the set $\{x\in X: f(x)=x\}$.

\begin{theorem}\label{orbKleene}Let $(X,\preceq)$ be a partially ordered set and let $f:X\rightarrow X$ be a mapping. Then the following are equivalent:

\begin{enumerate}

\item[(1)] $x^{\star}\in Fix(f)\neq \emptyset$.

\item[(2)]  There exists $x_{0}\in X$ such that 
\begin{enumerate}
\item[(2.1)] The sequence $(f^n(x_0))_{n\in\mathbb{N}^\star}$ is increasing in $(X,\preceq)$,
\item[(2.2)] $x^{\star}$ is the supremum of $(f^n(x_0))_{n\in\mathbb{N}^\star}$ and, thus, $x^{\star}\in \uparrow_{\preceq}x_{0}$,
\item[(2.3)] $f$ is orbitally $\preceq$-continuous at $x_{0}$.
\end{enumerate}
\item[(3)]  There exists $z_{0}\in X$ such that 
\begin{enumerate}
\item[(3.1)] $z_0\preceq f(z_0)$ in $(X,\preceq)$,
\item[(3.2)] $x^{\star}$ is the supremum of $(f^n(z_0))_{n\in\mathbb{N}^\star}$ and, thus, $x^{\star}\in \uparrow_{\preceq}z_{0}$,
\item[(3.3)] $f$ is orbitally $\preceq$-continuous at $z_{0}$.
\end{enumerate}
\end{enumerate}
\end{theorem}

\begin{proof} To show that $(1)\Rightarrow (2)$ it is sufficient to set $x^{\star}=x_{0}$ with $x^{\star}\in Fix(f)$. Furthermore, it is not hard to check that $(2)\Rightarrow (3)$. Indeed, if we take $z_0=x_0$, then (3) is satisfied, since $z_0\preceq f(z_0),$ due to the sequence $(f^nz_0))_{n\in\mathbb{N}^\star}$ is increasing in $(X,\preceq)$. So, it remains to prove that $(3)\Rightarrow (1)$. To this end, suppose that there exist $z_{0}\in X$ satisfying $(3.1), (3.2)$ and $(3.3)$. On the one hand, since $x^\star$ is the supremum of the sequence $(f^n(z_0))_{n\in\mathbb{N}^\star}$ in $(X,\preceq)$ and $z_0\preceq f(z_0)$, then $x^\star$ is the supremum $(f^{n+1}(z_0))_{n\in\mathbb{N}^\star}$. On the other hand, since $f$ is orbitally $\preceq$-continuous at $z_{0}$ we have that $f(x^{\star})$ is the supremum of $(f^{n+1}(z_0))_{n\in\mathbb{N}^\star}$ in $(X,\preceq)$. Hence $f(x^{\star})=x^{\star}$. \end{proof}

The next example shows that Theorem \ref{orbKleene} does not give, in general, the uniqueness of fixed point.

\begin{example}\label{NUFP}Consider the partially ordered set $([0,1],\leq)$ introduced in Example \ref{non_order_cont}. Let $f:[0,1]\rightarrow [0,1]$ be the mapping given by $f(x)=x$ for all $x\in [0,1]$. It is obvious that the sequence $(f^n(x_{0}))_{n\in\mathbb{N}^\star}$ is increasing in $([0,1],\leq)$, for all $x_{0}\in[0,1]$, and in addition, $x_{0}$ is the supremum of $(f^n(x_{0}))_{n\in\mathbb{N}^\star}$ in $([0,1],\leq)$. Moreover, $f$ is orbitally $\leq$-continuous at $x_{0}$, for all $x_{0}\in [0,1]$. Clearly, $Fix(f)=[0,1]$.
\end{example}


In the particular case in which the self-mapping is $\preceq$-continuous we get the following result.

\begin{corollary}\label{orbKleenecor}Let $(X,\preceq)$ be a partially ordered set and let $f:X\rightarrow X$ be a mapping. Assume that there exists $x_{0}\in X$ such that 
\begin{enumerate}
\item[(1)] $x_0\preceq f(x_0)$,
\item[(2)] $x^{\star}$ is the supremum of $(f^n(x_0))_{n\in\mathbb{N}^\star}$ and, thus, $x^{\star}\in \uparrow_{\preceq}x_{0}$,
\item[(3)] $f$ is $\preceq$-continuous.
\end{enumerate}
Then $x^{\star}\in Fix(f)\neq \emptyset$. Moreover, $x^{\star}\in \downarrow_{\preceq}y_{0}$ provided that $y_{0}\in X$ such that $x_{0}\preceq y_{0}$ and $f(y_{0})\preceq y_{0}$. Furthermore, $x^{\star}$ is the minimum of $Fix(f)\cap\uparrow_{\preceq}x_{0}$ in $(X,\preceq)$.
\end{corollary}

\begin{proof}Since $f$ is monotone and $x_{0}\preceq f(x_{0})$ we have that $(f^n(x_0))_{n\in\mathbb{N}^\star}$ is increasing in $(X,\preceq)$. Since $f$ is $\preceq$-continuous and $x_{0}\preceq f(x_{0})$ we have that $f$ is orbitally $\preceq$-continuous at $x_{0}$. Hence the existence of $x^{\star}\in X$ such that $x^{\star}\in Fix(f)$ is guaranteed by Theorem  \ref{orbKleene}.

Next we assume that there exists $y_{0}\in X$ such that $x_{0}\preceq y_{0}$ and that $f(y_{0})\preceq y_{0}$. Then $f^{n}(x_{0})\preceq f(y_{0})\preceq y_{0}$ for all $n\in\mathbb{N}$. It follows that $y_{0}$ is an upper bound of $(f^n(x_0))_{n\in\mathbb{N}^\star}$ in $(X,\preceq)$. Moreover, since $x^{\star}$ is the supremum of $(f^n(x_0))_{n\in\mathbb{N}^\star}$ in $(X,\preceq)$ we deduce that $x^{\star}\preceq y_{0}$. Whence we obtain that $x^{\star}\in \downarrow_{\preceq}y_{0}$.

It remains to prove that $x^{\star}$ is the minimum of $Fix(f)\cap\uparrow_{\preceq}x_{0}$ in $(X,\preceq)$. With this aim we suppose that there exists $y^{\star}\in Fix(f)\cap\uparrow_{\preceq}x_{0}$. As it was pointed out above $f$ is monotone and, thus, $f^{n}(x_{0})\preceq y^{\star}$. So, since $x^{\star}$ is the supremum of $(f^n(x_0))_{n\in\mathbb{N}^\star}$ we have that $x^{\star}\preceq y^{\star}$ as we claimed. \end{proof}

Taking into account Theorem \ref{orbKleene} we obtain the next result.

\begin{corollary}\label{orbdcop}Let $(X,\preceq)$ be a chain complete partially ordered set and let $f:X\rightarrow X$ be a mapping. Then the following are equivalent:

\begin{enumerate}

\item[(1)] $Fix(f)\neq \emptyset$.

\item[(2)]  There exists $x_{0}\in X$ such that 
\begin{enumerate}
\item The sequence $(f^n(x_0))_{n\in\mathbb{N}^\star}$ is increasing in $(X,\preceq)$,
\item $f$ is orbitally $\preceq$-continuous at $x_{0}$.
\end{enumerate}
\end{enumerate}
In addition, there exists $x^{\star}\in Fix(f)$ such that $x^{\star}$ is the supremum of the sequence $(f^n(x_0))_{n\in\mathbb{N}^\star}$ and, thus, $x^{\star}\uparrow_{\preceq}x_{0}$. 
\end{corollary}

\begin{proof} By the same arguments as in Theorem \ref{orbKleene} we have that $(1)\Rightarrow (2)$. To show that $(2)\Rightarrow (1)$, assume that there exists $x_0\in X$ satisfying (a) and (b). The fact that the partially ordered set $(X,\preceq)$ is chain complete provides the existence of $x^{\star}\in X$ such that $x^{\star}$ is the supremum of $(f^n(x_0))_{n\in\mathbb{N}^\star}$ and, thus, $x^{\star}\uparrow_{\preceq}x_{0}$. By Theorem \ref{orbKleene} we obtain that $x^{\star}\in Fix(f)$ and, hence, that $Fix(f)\neq \emptyset$.
\end{proof}

Combining Corollaries \ref{orbKleenecor} and \ref{orbdcop} we deduce the following one.

\begin{corollary}\label{orbdcop2}Let $(X,\preceq)$ be a chain complete partially ordered set and let $f:X\rightarrow X$ be a mapping. Assume that there exists $x_{0}\in X$ such that 
\begin{enumerate}
\item[(1)] $x_0\preceq f(x_0)$,
\item[(2)] $f$ is $\preceq$-continuous.
\end{enumerate} Then there exists $x^{\star}\in Fix(f)\neq \emptyset$. Moreover, $x^{\star}\in \downarrow_{\preceq}y_{0}$ provided that $y_{0}\in X$ such that $x_{0}\preceq y_{0}$ and $f(y_{0})\preceq y_{0}$. Furthermore, $x^{\star}$ is the minimum of $Fix(f)\cap\uparrow_{\preceq}x_{0}$ in $(X,\preceq)$.\end{corollary}

When the self-mapping is assumed to be only monotone (not $\preceq$-continuous), Theorem \ref{orbKleene} yields the following results which provide a bit more information about the fixed point than the aforesaid theorem and improves Corollary \ref{orbKleenecor}.

\begin{corollary}\label{orbmonot}Let $(X,\preceq)$ be a partially ordered set and let $f:X\rightarrow X$ be a monotone mapping. The following are equivalent:

\begin{enumerate}

\item[(1)] $x^{\star}\in Fix(f)\neq \emptyset$.

\item[(2)]  There exists $x_{0}\in X$ such that 
\begin{enumerate}
\item $x_{0}\preceq f(x_{0})$,
\item $x^{\star}$ is the supremum of $(f^n(x_0))_{n\in\mathbb{N}^\star}$ and, thus, $x^{\star}\in \uparrow_{\preceq}x_{0}$, 
\item $f$ is orbitally $\preceq$-continuous at $x_{0}$.
\end{enumerate}
\end{enumerate}
In addition, $x^{\star}\in \downarrow_{\preceq}y_{0}$ provided that $y_{0}\in X$ such that $y_{0}\in \uparrow_{\preceq}x_{0}$ and $f(y_{0})\preceq y_{0}$. Moreover, $x^{\star}$ is the minimum of $Fix(f)\cap\uparrow_{\preceq}x_{0}$ in $(X,\preceq)$.
\end{corollary}

\begin{proof}$(2)\Rightarrow (1)$. Since $x_{0}\preceq f(x_{0})$ and $f$ is monotone we have that the sequence $(f^n(x_0))_{n\in\mathbb{N}^\star}$ is increasing in $(X,\preceq)$. So all assumptions in the statement of Theorem \ref{orbKleene} are hold. Therefore, Theorem \ref{orbKleene} gives that there exists $x^{\star}\in Fix(f)$ which is the supremum of $(f^n(x_0))_{n\in\mathbb{N}^\star}$ and, thus, $x^{\star}\in \uparrow_{\preceq}x_{0}$.

The same arguments to those given in the proof of Corollary \ref{orbKleenecor} can be applied to conclude the remainder assertions in the statement of the result.

To prove that $(1)\Rightarrow (2)$ it is enough to take $x_{0}=x^{\star}$ with $x^{\star}\in Fix(f)$.

\end{proof}

The next example shows that we cannot omit the monotony of the self-mapping in the preceding result in order to guarantee that ``$x^{\star}\in \downarrow_{\preceq}y_{0}$ provided that $y_{0}\in X$ such that $y_{0}\in \uparrow_{\preceq}x_{0}$ and $f(y_{0})\preceq y_{0}$''.

\begin{example}\label{NM}Consider the partially ordered set $(X,\preceq_{X})$ and the self-mapping introduced in Example \ref{NM2}. It is clear that $0\in Fix(f)$. Corollary \ref{orbmonot} guarantees that there exists $x_{0}\in X$ ($x_{0}\in [0,1]$) such that $x_{0}\preceq_{X} f(x_{0})$, $0$ is the supremum of $(f^n(x_{0}))_{n\in\mathbb{N}^\star}$ and $f$ is orbitally $\preceq_{X}$-continuous at $x_{0}$. Moreover, it is obvious that $f(2)\preceq_{X} 2$ and $x_{0}\preceq_{X}2$ for all $x_{0}\in ]0,1]$. However, $0\not\preceq_{X}2$.  
\end{example}

The chain completeness of the partially ordered set allows to refine Corollary \ref{orbmonot} obtaining the result below.

\begin{corollary}\label{monotdcpo}Let $(X,\preceq)$ be a chain complete partially ordered set and let $f:X\rightarrow X$ be a monotone mapping. Then the following are equivalent:\begin{enumerate}
\item[(1)] $Fix(f)\neq \emptyset$.
\item[(2)]  There exists $x_{0}\in X$ such that 
\begin{enumerate}
\item $x_{0}\preceq f(x_{0})$,
\item $f$ is orbitally $\preceq$-continuous at $x_{0}$.
\end{enumerate}
\end{enumerate}
In addition, there exists $x^{\star}\in Fix(f)$ such that $x^{\star}$ is the supremum of the sequence $(f^n(x_0))_{n\in\mathbb{N}^\star}$ and, thus, $x^{\star}\uparrow_{\preceq}x_{0}$. Moreover, $x^{\star}\in \downarrow_{\preceq}y_{0}$ provided that $y_{0}\in X$ such that $x_{0}\preceq y_{0}$ and $f(y_{0})\preceq y_{0}$. Furthermore, $x^{\star}$ is the minimum of $Fix(f)\cap\uparrow_{\preceq}x_{0}$ in $(X,\preceq)$.
\end{corollary}

\begin{proof}$(1)\Rightarrow (2)$. It is sufficient to take $x^{\star}\in Fix(f)$ and set $x_{0}=x^{\star}$.

$(2)\Rightarrow (1)$. Since $f$ is monotone we have that the sequence $(f^n(x_0))_{n\in\mathbb{N}^\star}$ is increasing in $(X,\preceq)$. The chain completeness of $(X,\preceq)$ warranties the existence of the supremum $x^{\star}$ of $(f^n(x_0))_{n\in\mathbb{N}^\star}$ in $(X,\preceq)$ and so $x^{\star}\in \uparrow_{\preceq}x_{0}$. Besides, $x^{\star}\in Fix(f)$ by Corollary \ref{orbmonot}.

Similar arguments to those given in Corollary \ref{orbKleenecor} apply to show that $x^{\star}\in \downarrow_{\preceq}y_{0}$ provided that $y_{0}\in X$ such that $x_{0}\preceq y_{0}$ and $f(y_{0})\preceq y_{0}$ and to show that, in addition, $x^{\star}$ is the minimum of $Fix(f)\cap\uparrow_{\preceq}x_{0}$ in $(X,\preceq)$.
\end{proof}

Observe that Corollary \ref{monotdcpo} improves the celebrated Kleene fixed point theorem (see Theorem \ref{monotdcpo2}).


Let us recall that some distinguished partially ordered sets which play a central role in Computer Science are those that come from a quasi-metric space (see, for instance,  {\cite{JGL2013,SedaHit2}}). In the following we focus our attention on obtaining appropriate versions of the exposed results in those cases in which the partial order is induced by a quasi-metric. To this end, we recall a few notions about quasi-metric spaces that we will require later on.

Following \cite{Ku} (see also \cite{JGL2013}), a quasi-metric on a nonempty set $X$ is a function $d:X\times X\rightarrow $ $\mathbb{R}^{+}$ such that for
all $x,y,z\in X: $

$%
\begin{array}{ll}
\text{\textrm{(i)}} & d(x,y)=d(y,x)=0\Leftrightarrow x=y, \\
\text{\textrm{(ii)}} & d(x,z)\leq d(x,y)+d(y,z).%
\end{array}%
 $

Each quasi-metric $d$ on a set $X$ induces a $T_{0}$ topology $\tau(d) $ on\thinspace $X$ which has as a base the family of open $d$-balls $%
\{B_{d}(x,r):x\in X,$ $r>0\},$ where $B_{d}(x,r)=\{y\in X:d(x,y)<r\}$ for
all $x\in X$ and $r>0.$

A quasi-metric space is a pair $(X,d)$ such that $X$ is a nonempty set and
$d$ is a quasi-metric on $X.$

If $d$ is a quasi-metric on a set $X$, then the functions $d^{-1}$ and $d^{s}$ defined on $X\times X$ by $d^{-1}(x,y)=d(y,x)$ and $d^{s}(x,y)=\max \{d(x,y),d^{-1}(x,y)\}$ for all $x,y\in X$ are a quasi-metric and a metric on $X$, respectively.

Every quasi-metric space $(X,d)$ becomes a partially ordered set endowed with the specialization partial order $\preceq_{d}$. The specialization partial order $\preceq_{d}$ is defined on $X$ as follows: $x\preceq_{d}y \Leftrightarrow d(x,y)=0$  {(see \cite{JGL2013})}.

According to \cite{LRV2018}, a quasi-metric space $(X,d)$ is chain complete provided that the associated partially ordered set $(X,\preceq_{d})$ is chain complete. Clearly from the preceding results we get a sequence of corollaries when the partial order is assumed to be the specialization partial order coming from a quasi-metric. We only stress two of the aforementioned results, when the partial order matches up with the specialization one, because they will be of special interest later on.

\begin{corollary}\label{orbdcopCC}Let $(X,d)$ be a chain complete quasi-metric space and let $f:X\rightarrow X$ be a mapping.  Then the following are equivalent:
\begin{enumerate}

\item[(1)] $Fix(f)\neq \emptyset$.

\item[(2)]  There exists $x_{0}\in X$ such that 
\begin{enumerate}
\item The sequence $(f^n(x_0))_{n\in\mathbb{N}^\star}$ is increasing in $(X,\preceq_{d})$,
\item $f$ is orbitally $\preceq_{d}$-continuous at $x_{0}$.
\end{enumerate}
\end{enumerate}
In addition, there exists $x^{\star}\in Fix(f)$ such that $x^{\star}$ is the supremum of the sequence $(f^n(x_0))_{n\in\mathbb{N}^\star}$ and, thus, $x^{\star}\uparrow_{\preceq_{d}}x_{0}$. 
\end{corollary}

Notice that the preceding result comes from Corollary \ref{orbdcop}. If in addition, we demand monotony on the mapping we obtain the next corollary which is derived from Corollary \ref{monotdcpo}.

\begin{corollary}\label{monotdcpoCC}Let $(X,d)$ be a chain complete quasi-metric space and let $f:X\rightarrow X$ be a monotone mapping. Then the following are equivalent:

\begin{enumerate}
\item[(1)] $Fix(f)\neq \emptyset$.
\item[(2)]  There exists $x_{0}\in X$ such that 
\begin{enumerate}
\item $x_{0}\preceq_{d} f(x_{0})$,
\item $f$ is orbitally $\preceq_{d}$-continuous at $x_{0}$.
\end{enumerate}
\end{enumerate}
In addition, there exist $x^{\star}\in Fix(f)$ such that $x^{\star}$ is the supremum of the sequence $(f^n(x_0))_{n\in\mathbb{N}^\star}$ and, thus, $x^{\star}\uparrow_{\preceq_{d}}x_{0}$. Moreover, $x^{\star}\in \downarrow_{\preceq_{d}}y_{0}$ provided that $y_{0}\in X$ such that $x_{0}\preceq_{d} y_{0}$ and $f(y_{0})\preceq_{d} y_{0}$. Furthermore, $x^{\star}$ is the minimum of $Fix(f)\cap\uparrow_{\preceq_{d}}x_{0}$ in $(X,\preceq_{d})$.
\end{corollary}

It must be stressed that Corollary \ref{monotdcpoCC} improves Theorem 7 in \cite{LRV2018}, since it gives a characterization about the existence of fixed point. Notice that the aforesaid Theorem 7 only proves the implication $(2)\Rightarrow (1)$ when the self-mapping is $\preceq_{d}$-continuous. Besides, Corollary \ref{monotdcpoCC} yields information about the fixed point in the particular case in which there exists ``$y_{0}\in X$ such that $x_{0}\preceq_{d} y_{0}$ and $f(y_{0})\preceq_{d} y_{0}$'' and such an information is not provided by Theorem 7.

It seems natural to wonder whether there are a wide number of examples of chain complete quasi-metric spaces $(X,d)$, or on the contrary if it is strange to find instances of this type of spaces. The next results answer the posed question affirmative, i.e., showing that the so-called $\preceq_{d}$-complete (in the sense of \cite{LRV2018}) provide a wide class of quasi-metric spaces that satisfy the aforesaid property (see Propositions \ref{completeness} and \ref{usefullemma} below). Before introducing the announced result let us recall that a quasi-metric space $(X,d)$ is $\preceq_{d}$-complete provided that each increasing sequence $(x_{n})_{n\in\mathbb{N}}$ in $(X,\preceq_{d})$ converges with respect to $\tau(d^{s})$.

In view of the above introduced notion we show that there are a wide class of quasi-metric spaces which are $\preceq_{d}$-complete. To this end, let us recall a few appropriate notions of completeness that arise in a natural way in the quasi-metric framework.

According to \cite{Reilly}, a sequence $(x_{n})_{n\in\mathbb{N}}$ in a quasi-metric space $(X,d)$ is said to be right (left) K-Cauchy if, given $\varepsilon>0$, there exists $n_{0}\in\mathbb{N}$ such that $d(x_{m},x_{n})<\varepsilon$ ($d(x_{n},x_{m})<\varepsilon$) for all $m\geq n\geq n_{0}$. A quasi-metric space $(X,d)$ is said to be right K-sequentially complete provided that every right K-Cauchy sequence converges with respect to $\tau(d)$. Following \cite{Cobzas} (see also \cite{KSch}), a quasi-metric space $(X,d)$ is left (right) Smyth complete provided that every left (right) K-Cauchy sequence converges with respect to $\tau(d^{s})$. On account of \cite{Ma}, a quasi-metric space $(X,d)$ is called weightable provided the existence of a function $w_{d}:X\rightarrow \mathbb{R}^{+}$ such that 
$$d(x,y)+w_{d}(x)=d(y,x)+w_{d}(y)$$ for all $x,y\in X$. Finally, a quasi-metric space $(X,d)$ is said to be bicomplete if the induced metric space $(X,d^{s})$ is complete (see, for instance, \cite{Ku}).

Next we show that all preceding classes of ``complete'' quasi-metric spaces are instances of $\preceq_{d}$-complete quasi-metric spaces. To this end, we count with the help of Lemma \ref{least} whose proof we omit because it was given in \cite{LRV2018}.

\begin{lemma}\label{least}Let $(X,d)$ be a quasi-metric space. If $x\in X$ and $(x_{n})_{n\in\mathbb{N}}$ is an increasing sequence in $(X,\preceq_{d})$ which converges to $x$ with respect to $\tau(d^{s})$, then $x$ is the supremum of $(x_{n})_{n\in\mathbb{N}}$ in $(X,\preceq_{d})$. \end{lemma}

\begin{proposition}\label{completeness}Let $(X,d)$ be a quasi-metric space such that one of the following assertions holds:
\begin{enumerate}
\item $(X,d)$ is left Smyth complete,
\item $(X,d^{-1})$ is right Smyth complete,
\item $(X,d)$ is weightable and bicomplete.
\end{enumerate} Then $(X,d)$ is $\preceq_{d}$-complete.
\end{proposition}

\begin{proof}\begin{enumerate}
\item Let $(x_{n})_{n\in\mathbb{N}}$ be an increasing sequence in $(X,\preceq_{d})$. Then there exists $n_{0}\in\mathbb{N}$ such that $d(x_{n},x_{m})=0$ for all $m\geq n\geq n_{0}$. Thus $d(x_{n},x_{m})=0$ for all $m\geq n\geq n_{0}$. It follows that the sequence $(x_{n})_{n\in\mathbb{N}}$ is left K-Cauchy in $(X,d)$. Since the quasi-metric space $(X,d)$ is left Smyth complete we deduce the existence of $x\in X$ such that $(x_{n})_{n\in\mathbb{N}}$ converges to $x$ with respect to $\tau(d^{s})$. By Lemma \ref{least} we obtain that $x$ is the supremum of $(x_{n})_{n\in\mathbb{N}}$ in $(X,\preceq_{d})$. 

\item Let $(x_{n})_{n\in\mathbb{N}}$ be an increasing sequence in $(X,\preceq_{d})$. Then there exists $n_{0}\in\mathbb{N}$ such that $d(x_{n},x_{m})=0$ for all $m\geq n\geq n_{0}$. Hence we have that $d^{-1}(x_{m},x_{n})=0$ for all $m\geq n\geq n_{0}$. Since the quasi-metric space $(X,d^{-1})$ is right Smyth complete we deduce the existence of $x\in X$ such that $(x_{n})_{n\in\mathbb{N}}$ converges to $x$ with respect to $\tau(d^{s})$. By Lemma \ref{least} we obtain that $x$ is the supremum of $(x_{n})_{n\in\mathbb{N}}$ in $(X,\preceq_{d})$.

\item  On account of \cite{Ku}, every weightable bicomplete quasi-metric space is always left Smyth complete. 
\end{enumerate}

\end{proof}

The following result states that every quasi-metric, which is complete in any sense of Proposition \ref{completeness}, is chain complete. 

\begin{proposition}\label{usefullemma}Let $(X,d)$ be a $\preceq_{d}$-complete quasi-metric space. Then $(X,\preceq_{d})$ is chain complete.
\end{proposition}

\begin{proof}Let $(x_{n})_{n\in\mathbb{N}}$ be an increasing sequence in $(X,\preceq_{d})$. Since the quasi-metric space $(X,d)$ is $\preceq_{d}$-complete we have that there exists $x\in X$ such that $(x_{n})_{n\in\mathbb{N}}$ converges to $x$ with respect to $\tau(d^{s})$. By Lemma \ref{least} we deduce that $x$ is the supremum of $(x_{n})_{n\in\mathbb{N}}$ in $(X,\preceq_{d})$. It follows that $(X,\preceq_{d})$ is chain complete. \end{proof}

From Corollary \ref{orbdcopCC} we deduce the next two results.

\begin{corollary}\label{apply1}Let $(X,d)$ be a $\preceq_{d}$-complete quasi-metric space and let $f:X\rightarrow X$ be a mapping. Then the following are equivalent:

\begin{enumerate}

\item[(1)] $Fix(f)\neq \emptyset$.

\item[(2)]  There exists $x_{0}\in X$ such that 
\begin{enumerate}
\item The sequence $(f^n(x_0))_{n\in\mathbb{N}^\star}$ is increasing in $(X,\preceq_{d})$,
\item $f$ is orbitally $\preceq_{d}$-continuous at $x_{0}$.
\end{enumerate}
\end{enumerate}
In addition, there exists $x^{\star}\in Fix(f)$ such that $x^{\star}$ is the supremum of the sequence $(f^n(x_0))_{n\in\mathbb{N}^\star}$ and, thus, $x^{\star}\uparrow_{\preceq_{d}}x_{0}$. 
\end{corollary}

\begin{proof}By Proposition \ref{usefullemma} we have that the partially ordered set $(X,\preceq_{d})$ is chain complete. Applying Corollary \ref{orbdcopCC}  we obtain the desired conclusions. 

\end{proof}

\begin{corollary}\label{completeness1}Let $(X,d)$ be a quasi-metric space such that one of the following assertions holds:
\begin{enumerate}
\item $(X,d)$ is left Smyth complete,
\item $(X,d^{-1})$ is right Smyth complete,
\item $(X,d)$ is weightable and bicomplete.
\end{enumerate} Let $f:X\rightarrow X$ be a mapping. Then the following are equivalent:

\begin{enumerate}

\item[(1)] $Fix(f)\neq \emptyset$.

\item[(2)]  There exists $x_{0}\in X$ such that 
\begin{enumerate}
\item The sequence $(f^n(x_0))_{n\in\mathbb{N}^\star}$ is increasing in $(X,\preceq_{d})$,
\item $f$ is orbitally $\preceq_{d}$-continuous at $x_{0}$.
\end{enumerate}
\end{enumerate}
In addition, there exists $x^{\star}\in Fix(f)$ such that $x^{\star}$ is the supremum of the sequence $(f^n(x_0))_{n\in\mathbb{N}^\star}$ and, thus, $x^{\star}\uparrow_{\preceq_{d}}x_{0}$.
\end{corollary}

From Corollaries \ref{monotdcpoCC} and \ref{completeness1} we derive the next two results that will play a central role in our subsequent discussion.

\begin{corollary}\label{apply2}Let $(X,d)$ be a $\preceq_{d}$-complete quasi-metric space and let $f:X\rightarrow X$ be a monotone mapping. Then the following are equivalent:

\begin{enumerate}
\item[(1)] $Fix(f)\neq \emptyset$.
\item[(2)]  There exists $x_{0}\in X$ such that 
\begin{enumerate}
\item $x_{0}\preceq_{d} f(x_{0})$,
\item $f$ is orbitally $\preceq_{d}$-continuous at $x_{0}$.
\end{enumerate}
\end{enumerate}
In addition, there exists $x^{\star}\in Fix(f)$ such that $x^{\star}$ is the supremum of the sequence $(f^n(x_0))_{n\in\mathbb{N}^\star}$ and, thus, $x^{\star}\uparrow_{\preceq_{d}}x_{0}$. Moreover, $x^{\star}\in \downarrow_{\preceq_{d}}y_{0}$ provided that $y_{0}\in X$ such that $x_{0}\preceq_{d} y_{0}$ and $f(y_{0})\preceq_{d} y_{0}$. Furthermore, $x^{\star}$ is the minimum of $Fix(f)\cap\uparrow_{\preceq_{d}}x_{0}$ in $(X,\preceq_{d})$.
\end{corollary}

\begin{proof}By Proposition \ref{usefullemma} we have that the partially ordered set $(X,\preceq_{d})$ is chain complete. Applying Corollary \ref{monotdcpoCC} we obtain the desired conclusions. 

\end{proof}

\begin{corollary}\label{completeness2}Let $(X,d)$ be a quasi-metric space such that one of the following assertions holds:
\begin{enumerate}
\item $(X,d)$ is left Smyth complete,
\item $(X,d^{-1})$ is right Smyth complete,
\item $(X,d)$ is weightable and bicomplete.
\end{enumerate} Let $f:X\rightarrow X$ be a monotone mapping. Then the following are equivalent:

\begin{enumerate}
\item[(1)] $Fix(f)\neq \emptyset$.
\item[(2)]  There exists $x_{0}\in X$ such that 
\begin{enumerate}
\item $x_{0}\preceq_{d} f(x_{0})$,
\item $f$ is orbitally $\preceq_{d}$-continuous at $x_{0}$.
\end{enumerate}
\end{enumerate}
In addition, there exists $x^{\star}\in Fix(f)$ such that $x^{\star}$ is the supremum of the sequence $(f^n(x_0))_{n\in\mathbb{N}^\star}$ and, thus, $x^{\star}\uparrow_{\preceq_{d}}x_{0}$. Moreover, $x^{\star}\in \downarrow_{\preceq_{d}}y_{0}$ provided that $y_{0}\in X$ such that $x_{0}\preceq_{d} y_{0}$ and $f(y_{0})\preceq_{d} y_{0}$. Furthermore, $x^{\star}$ is the minimum of $Fix(f)\cap\uparrow_{\preceq_{d}}x_{0}$ in $(X,\preceq_{d})$.
\end{corollary}

\section{The application}

In 1995, M.P. Schellekens developed a new mathematical method to provide the asymptotic upper bounds of those algorithms whose running time of computing satisfies a recurrence equation (see \cite{Schellekens}). This method is based on the use of the so-called complexity space. Let us recall that the complexity space is the quasi-metric space $(\mathcal{C},d_{\mathcal{C}})$ where 

$$\mathcal{C}=\{f:\mathbb{N}\rightarrow \mathbb{R}^{+}: \sum_{n=1}^{\infty} 2^{-n}f(n)<\infty\}$$ and the quasi-metric $d_{\mathcal{C}}$ is given by 
$$d_{\mathcal{C}}(f,g)=\sum_{n=1}^{\infty} 2^{-n}\left(\max\left(\frac{1}{g(n)}-\frac{1}{f(n)},0 \right)\right).$$ 

On account of \cite{Schellekens}, each algorithm $A$ can be associated to a function $f_{A}\in \mathcal{C}$ such that $f_{A}(n)$ represents the time taken by $A$ to solve the problem for which $A$ has been designed when the size of input data is $n\in\mathbb{N}$. The mappings belonging to $\mathcal{C}$ were called complexity functions in \cite{Schellekens}.

Observe that the condition ``$\sum_{n=1}^{\infty} 2^{-n}f(n)<\infty$'' which is used to define $\mathcal{C}$ is not restrictive, since it is held by every computable algorithm, i.e., it is fulfilled by all algorithms $B$ with $f_{B}(n)\leq 2^{n}$ for all $n\in \mathbb{N}$. Moreover, the value $d_{\mathcal{C}}(f_{A},f_{B})$ can be understood as the relative progress made in lowering the complexity by replacing any algorithm $A$ with complexity function $f_{A}$ by any algorithms $B$ with complexity function $f_{B}$. Thus, the condition $d_{\mathcal{C}}(f_{A},f_{B})=0$ (or, equivalently, $f_{A}\preceq_{d_{\mathcal{C}}} f_{B}$) can be interpreted as the algorithm $A$ is at least as efficient as the algorithm $B$, since $d_{\mathcal{C}}(f_{A},f_{B})=0 \Leftrightarrow f_{A}(n)\leq f_{B}(n)$ for all $n\in\mathbb{N}$.

Notice that, given $g\in\mathcal{C}$, $d_{\mathcal{C}}(f_{A},g)=0$ implies that $f_{A}\in \mathcal{O}(g)$, where $\mathcal{O}(g)=\{f\in\mathcal{C}:\mbox{ there exists } c\in \mathbb{R}^{+} \mbox{ and } n_{0}\in \mathbb{N} \mbox{ with }	f_{A}(n)\leq c g(n) $ for all $ n\geq n_{0}\}$.  According to \cite{BB1988}, when the precise information about the running time of computing $f_{A}$ of an algorithm $A$ is not known, the fact that $f_{A}\in \mathcal{O}(g)$ yields an asymptotic upper bound of the time taken by $A$ in order to solve the problem under consideration. It must be stressed that the condition $d_{\mathcal{C}}(g,f_{A})=0$ can be also interpreted as $f_{A}\in \Omega(g)$, where $\Omega(g)=\{f\in\mathcal{C}:\mbox{ there exists } c\in \mathbb{R}^{+} \mbox{ and } n_{0}\in \mathbb{N} \mbox{ with }	cg(n)\leq f(n) \mbox{ for all } n\geq n_{0}\}$. Of course, from a computational viewpoint the fact that $f_{A}\in \Omega(g)$ provides that the mapping $g$ gives an asymptotic lower bound of the running time of computing of the algorithm $A$.

Observe that the asymmetry of $d_{\mathcal{C}}$ plays a central role in order to provide information about the increase in complexity whenever an algorithm is replaced by another one. Clearly, a metric would be able to yield information on the increase but it, however, will not reveal which algorithm is more efficient.

The utility of the complexity space $(\mathcal{C},d_{\mathcal{C}})$ was shown by Schellekens in \cite{Schellekens}, where he gave an alternative proof of the fact that the Mergesort has optimal asymptotic average running time of computing, i.e., $f_{M}\in \mathcal{O}(f_{\log})\cap\Omega(f_{\log})$, where $f_{M}$ represents the running time of the Mergesort and $f_{\log}\in \mathcal{C}$ such that $f_{log}(1)=c$ ($c\in \mathbb{R}^{+}$) and $f_{log}(n)=n\log_{2}(n)$ for all $n\in\mathbb{N}$ with $n>1$.  To achieve the mentioned target, Schellekens developed a technique based on the use of the celebrated Banach fixed point theorem. The aforesaid fixed point technique was applied to analyze those algorithms whose running time of computing satisfies a Divide and Conquer recurrence equation. Let us recall briefly that a Divide and Conquer recurrence equation is given as follows (see \cite{BB1988,Schellekens} for a detailed discussion):

\begin{equation}
T(n)=\left\{\begin{array}{ll}
c& \text{ if } n=1,\\
aT(\frac{n}{b})+h(n)& \text{ if } n\in \mathbb{N}_{b},\\ 
\end{array}\right. \label{DCR} 
\end{equation}  where $\mathbb{N}_{b}=\{b^{k}:k\in \mathbb{N}\}$, $c\in \mathbb{R}^{+}$, $a,b\in \mathbb{N}$ with $a,b>1$ and $h\in \mathcal{C}$ with $h(n)<\infty$ for all $n\in \mathbb{N}$.

Set $\mathcal{C}_{b,c}=\{f\in \mathcal{C}: f(1)=c \text{ and } f(n)=\infty \text{ for all } n\in \mathbb{N}\setminus \mathbb{N}_{b} \text{ with } n>1\}$. It is clear that a mapping $f\in \mathcal{C}_{b,c}$ is a solution to the recurrence equation (\ref{DCR}) if and only if $f$ is a fixed point of the mapping $\Phi_{T}:\mathcal{C}_{b,c}\rightarrow \mathcal{C}_{b,c}$ associated with the recurrence equation (\ref{DCR}) and given by 

\begin{equation}
\Phi_{T}(f)(n)=\left\{
\begin{array}{ll}
c & \text{\textrm{if }}n=1, \\
af(\frac{n}{b})+h(n) & \text{\textrm{if }} n\in \mathbb{N}_{b}, \\
\infty & \text{ otherwise} ,
\end{array}
\right.  \label{funct1}
\end{equation} for all $f\in \mathcal{C}_{b,c}.$

Concretely, the fixed point technique introduced by Schellekens is given by the following result:
\begin{theorem}\label{SchT}The quasi-metric space $(\mathcal{C}_{b,c},d_{\mathcal{C}})$ is left Smyth complete and the mapping $\Phi_{T}$ satisfies that $d_{\mathcal{C}}(\Phi_{T}(f),\Phi_{T}(g))\leq \frac{1}{2} d_{\mathcal{C}}(f,g)$ for all $f,g\in \mathcal{C}_{b,c}$. Thus, a Divide and Conquer recurrence of the form (\ref{DCR}) has a unique solution $f_{T}\in \mathcal{C}_{b,c}$. Moreover, the following assertions hold:
\begin{enumerate}
\item If there exists $g\in \mathcal{C}_{b,c}$ such that $g\preceq_{d_{\mathcal{C}}} \Phi_{T}(g)$, then $f_{T}\in \Omega(g)$.
\item If there exists $g\in \mathcal{C}_{b,c}$ such that $\Phi_{T}(g)\preceq_{d_{\mathcal{C}}} g$, then $f_{T}\in \mathcal{O}(g)$.
\end{enumerate}
\end{theorem}

The technique introduced by the above result was tested and illustrated successfully with the following particular case of the recurrence equation (\ref{DCR}):

\begin{equation}
T_{M}(n)=\left\{\begin{array}{ll}
c& \text{ if } n=1,\\
2T_{M}(\frac{n}{2})+\frac{n}{2}& \text{ if } n\in \mathbb{N}_{2},\\ 
\end{array}\right. \label{DCRMerge} 
\end{equation}  where $c\in \mathbb{R}^{+}$. Therefore Schellekens proved that the mapping $\Phi_{T_{M}}:\mathcal{C}_{2,c}\rightarrow \mathcal{C}_{2,c}$, defined by 

\begin{equation}
\Phi_{T_{M}}(f)(n)=\left\{
\begin{array}{ll}
c & \text{\textrm{if }}n=1, \\
2f(\frac{n}{2})+\frac{n}{2} & \text{\textrm{if }} n\in \mathbb{N}_{2}, \\
\infty & \text{ otherwise} ,
\end{array}
\right.  \label{funct1M}
\end{equation} for all $f\in \mathcal{C}_{2,c}$, satisfies the following: $g_{1}\preceq_{d_{\mathcal{C}}} \Phi_{T_{M}}(g_{1})$ and $\Phi_{T_{M}}(g_{2})\preceq_{d_{\mathcal{C}}} g_{2}$ for any $g\in \mathcal{C}_{2,c}$ if and only if $g_{1}=g_{2}$ and they are defined by

\begin{equation}
g(n)=\left\{
\begin{array}{ll}
c & \text{\textrm{if }}n=1, \\
\frac{1}{2} n\log_{2}(n) & \text{\textrm{if }} n\in \mathbb{N}_{2}, \\
\infty & \text{ otherwise} ,
\end{array}
\right.  \label{funct1M2}
\end{equation}

In \cite{RPV,RomVal2013}, the technique provided by Theorem \ref{SchT} was extended to those cases in which the recurrence equation associated to the running time of computing is of the type below:

\begin{equation}
T(n)=\left\{
\begin{array}{ll}
c_{n} & \text{\textrm{if }} 1\leq n\leq k \\
\sum_{i=1}^{k}a_{i}T(n-i)+h(n) & \text{\textrm{if }}n> k
\end{array}
\right. ,  \label{genreclinear}
\end{equation}
where $h\in \mathcal{C}$ such that $h(n)<\infty $ for all $n\in
\mathbb{N}$, $k\in\mathbb{N}$, $c_{i},a_{i}\in\mathbb{R}^{+}$ with $a_{i}\geq 1$ for all
$1\leq i\leq k$.

Observe that the recurrence equations of type (\ref{DCR}) can be recovered from those of type
(\ref{genreclinear}). In fact, the former recurrence equations can be transformed into one of the following type
\begin{equation}
S(m)=\left\{
\begin{array}{ll}
c & \text{\textrm{if }} m=1 \\
aS(m-1)+r(m) & \text{\textrm{if }}m> 1
\end{array}
\right. ,  \label{reclinearDC}
\end{equation}
where $S(m)=T(b^{m-1})$ and $r(m)=h(b^{m-1})$ for all
$m\in\mathbb{N}$. (Recall that $\mathbb{N} _{b}=\{b^{k}:k\in
\mathbb{N}\}$ with $b\in\mathbb{N} $ and $b>1$).

The asymptotic lower and upper bounds for a few celebrated algorithms, like Quicksort, Hanoi, Largetwo and Fibonnacci (see \cite{Cull,BB1988}), whose running time of computing holds the recurrence equation (\ref{genreclinear}), were discussed by means of appropriate versions of the technique exposed in Theorem \ref{SchT} and, thus, by means of the Banach fixed point theorem. Notice that in such versions the unique thing to be proved, additionally to original Schellekens' proof, was the contractive character of the mapping $\Phi_{T}: \mathcal{C}_{c_{1}\ldots, c_{k}}\rightarrow \mathcal{C}_{c_{1}\ldots, c_{k}}$ associated to the recurrence equation (\ref{genreclinear}) and the left Smyth completeness of the subset $\mathcal{C}_{c_{1}\ldots, c_{k}}$ with respect to $\tau(d^{s}_{\mathcal{C}})$ , where $\mathcal{C}_{c_{1}\ldots, c_{k}}=\{f\in \mathcal{C}: f(i)=c_{i} \mbox{ for all } 1\leq i\leq k\}$ and

\begin{equation}
\Phi_{T}(f)(n)=\left\{
\begin{array}{ll}
c_{i} & \text{\textrm{if }} 1\leq i\leq k, \\
\sum_{i=1}^{k}a_{i}f(n-i)+h(n)  & \text{\textrm{if }} n>k, \\
\end{array}
\right.  \label{genreclinear2}
\end{equation} for all $f\in \mathcal{C}_{c_{1}\ldots,c_{k}}$. Thus the technique introduced in Theorem \ref{SchT} was extended to the new case as follows:

\begin{theorem}\label{SchT2}The quasi-metric space $(\mathcal{C}_{c_{1}\ldots, c_{k}},d_{\mathcal{C}})$ is left Smyth complete and the mapping $\Phi_{T}$ given by (\ref{genreclinear2}) satisfies that $$ d_{\mathcal{C}}(\Phi_{T}(f),\Phi_{T}(g))\leq \left(\max_{1\leq i\leq k}\frac{1}{a_{i}}\right) \left( \frac{2^{k}-1}{2^{k}}\right) d_{\mathcal{C}}(f,g)$$ for all $f,g\in \mathcal{C}_{c_{1},\ldots,c_{k}}$. Thus, an algorithm whose running time of computing holds a recurrence equation of the form (\ref{genreclinear}) has a unique solution $f_{T}\in \mathcal{C}_{c_{1}\ldots, c_{k}}$. Moreover, the following assertions hold:
\begin{enumerate}
\item If there exists $g\in\mathcal{C}_{c_{1}\ldots, c_{k}}$ such that $g\preceq_{d_{\mathcal{C}}} \Phi_{T}(g)$, then $f_{T}\in \Omega(g)$.
\item If there exists $g\in \mathcal{C}_{c_{1}\ldots, c_{k}}$ such that $\Phi_{T}(g)\preceq_{d_{\mathcal{C}}} g$, then $f_{T}\in \mathcal{O}(g)$.
\end{enumerate}
\end{theorem}

Notice that, by means of the transformation given by (\ref{reclinearDC}), Theorem \ref{SchT} can be retrieved from Theorem \ref{SchT2}.

It must be stressed that the uniqueness of solution to the recurrence equations (or equivalently the uniqueness of fixed point of the mapping $\Phi_{T}$) under consideration in Theorems \ref{SchT} and \ref{SchT2} is guaranteed by the left Smyth completeness and the Banach fixed point theorem (we refer the reader to \cite{Schellekens} for a detailed discussion). However, from a complexity analysis viewpoint, it is not necessary to debate about the uniqueness of the solution because the theory of finite difference equations provides such a uniqueness for the so-called initial value problems (see, for instance, Theorem 3.1.1 in \cite{Cull}). So, the really novel and interesting about the techniques introduced by Theorems \ref{SchT} and  \ref{SchT2} is exactly  the possibility of studying the asymptotic behavior of the solutions via fixed point arguments which differs from the classical difference equation approach (see, again, \cite{Cull}).

Inspired, in part, by the fact already exposed,  L.M. Garc\'ia-Raffi, S. Romaguera and Schellekens provided a mathematical method for asymptotic complexity analysis of algorithms which is not based on the use of the Banach fixed point theorem, or equivalently of Theorems \ref{SchT} and \ref{SchT2}, in \cite{GRRS2008}. Concretely they provided, by means of fixed point techniques and the use of increasing sequences of complexity functions, asymptotic upper bounds for the running time of computing of the so-called Probabilistic Divide and Conquer algorithms (see \cite{Knuth1973} for a detailed discussion of this type of algorithms).

Let us recall that the running time of computing of Probabilistic Divide and Conquer algorithms satisfies the following recurrence equation:

\begin{equation}
T(n)=\left\{
\begin{array}{ll}
c_{n} & \text{\textrm{if }} 1\leq n<k \\
\sum_{i=1}^{n-1}v_{i}(n)T(i)+h(n) & \text{\textrm{if }}n\geq  k
\end{array}
\right. ,  \label{probabilistic}
\end{equation}
where $h\in \mathcal{C}$ such that $h(n)<\infty $ for all $n\in
\mathbb{N}$, $k\in\mathbb{N}$ such that $k\geq 2$ and $c_{i}\in\mathbb{R}^{+}$ for all $1\leq i<k$. Moreover, $(v_{i})_{i\in\mathbb{N}}$ is a sequence of positive mappings defined on $\mathbb{N}$ in such a way that there exists $K\in\mathbb{R}^{+}$ with $K>0$ satisfying that $\sum_{i=1}^{n-1}v_{i}(n)\leq K$.

 To get asymptotic upper bounds of the running time in those cases in which the recurrence equation (\ref{probabilistic}) is under consideration the next auxiliary result was key and it was proved in \cite{GRRS2008}.

\begin{proposition}\label{useful}Let $\mathcal{R}\subseteq \mathcal{C}$ such that $(\mathcal{R},d_{\mathcal{C}})$ is left Smyth complete. Let $\Phi:\mathcal{R} \rightarrow \mathcal{R}$ be a monotone mapping with respect to $\preceq_{d_{\mathcal{C}}}$. If there exists $g\in \mathcal{R}$ such that $g\preceq_{d_{\mathcal{C}}} \Phi(g)$, then there exists $f\in \mathcal{R}$ such that the sequence $(\Phi^{n}(g))_{n\in\mathbb{N}^\star}$ converges to $f$ with respect to $\tau(d^{s}_{\mathcal{C}})$ and, in addition, $f$ is an upper bound of $(\Phi^{n}(g))_{n\in\mathbb{N}^\star}$ in $(X,\preceq_{d_{\mathcal{C}}})$.
\end{proposition}

A specific method to provide the aforementioned asymptotic upper bounds for the solution to recurrence equations of type (\ref{probabilistic}) was proved using Proposition \ref{useful} in \cite{GRRS2008}. Concretely, it was given the result below.

\begin{theorem}\label{R}Let $k\in\mathbb{N}$ with $k\geq 2$ and let $\mathcal{C}_{c_{1},\ldots,c_{k-1}}$ be the subset of $\mathcal{C}$ given by $\mathcal{C}_{c_{1},\ldots,c_{k-1}}=\{f\in \mathcal{C}:f(i)=c_{i} \text{ for all } 1\leq i<k\}$. Define the mapping $\Phi_{T}:\mathcal{C}_{c_{1},\ldots,c_{k-1}}\rightarrow \mathcal{C}_{c_{1},\ldots,c_{k-1}}$ by 

\begin{equation}
\Phi_{T}(f)(n)=\left\{
\begin{array}{ll}
c_{n} & \text{\textrm{if }} 1\leq n<k \\
\sum_{i=1}^{n-1}v_{i}(n)f(i)+h(n) & \text{\textrm{if }} n\geq k \\
\end{array}
\right.  ,\label{RecProb}
\end{equation} for all $f\in \mathcal{C}_{c_{1},\ldots,c_{k-1}}$. Then the following assertions hold:
\begin{enumerate}
\item The quasi-metric space $(\mathcal{C}_{c_{1},\ldots,c_{k-1}},d_{\mathcal{C}})$ is left Smyth complete 
\item The mapping $\Phi_{T}$ is monotone with respect to $\preceq_{d_{\mathcal{C}}}$ and there exists $f_{T}\in \mathcal{C}_{c_{1},\ldots,c_{k-1}}$ such that $Fix(\Phi_{T})=\{f_{T}\}$. So $f_{T}$ is the unique solution to the recurrence equation (\ref{probabilistic}). 
\item If there exists $f\in \mathcal{C}_{c_{1},\ldots,c_{k-1}}$ such that $\Phi(f)\preceq_{d_{\mathcal{C}}} f$, then $f_{T}\in \mathcal{O}(f)$.
\end{enumerate}
\end{theorem}

The advantage of the method exposed in the preceding result is given by the fact that it makes use of the Banach fixed point theorem. However, the aforesaid method has been designed specifically for Probabilistic Divide and Conquer algorithms. Observe, in addition, that the uniqueness of solution to the recurrence equation (\ref{probabilistic}) was warrantied by means of induction techniques in \cite{GRRS2008}, i.e., following the aforesaid classical techniques from finite difference equations. Motivated by this fact we show that the theory exposed in Section \ref{Results} provides a general framework for discussing asymptotic bounds (upper  and lower) of the complexity of algorithms in such a way that both mathematical methods for such a purpose given in Theorems \ref{SchT}, \ref{SchT2} and \ref{R} can be retrieved as a particular case. In particular we can state the below method for asymptotic complexity analysis of algorithms. Notice that such a method does not deal with uniqueness since that's what the theory of finite difference equation guarantees.

\begin{theorem}\label{CompTech}Let $\mathcal{R}\subseteq \mathcal{C}$ such that $(\mathcal{R},\preceq_{d_{\mathcal{C}}})$ is chain complete.  Let $\Phi:\mathcal{R} \rightarrow \mathcal{R}$ be a monotone mapping. If there exist $f,g\in \mathcal{R}$ such that the following assertions hold:
\begin{enumerate}
\item $g\preceq_{d_{\mathcal{C}}} \Phi(g)$ and $\Phi$ is orbitally $\preceq_{d_{\mathcal{C}}}$-continuous at $g$,
\item $g\preceq_{d_{\mathcal{C}}} f$ and  $\Phi(f)\preceq_{d_{\mathcal{C}}} f$.
\end{enumerate} Then there exists $f^{\star}\in \mathcal{R}$ such that $f^{\star}\in Fix(\Phi)$ and $f^{\star}\in \Omega(g)\cap \mathcal{O}(f)$.
\end{theorem}

\begin{proof}By Corollary \ref{monotdcpoCC} we deduce that $Fix(\Phi)\neq\emptyset$ and that there exists $f^{\star}\in Fix(\Phi)$ such that $f^{\star}\in \Omega(g)\cap \mathcal{O}(f)$.  \end{proof}

\begin{corollary}\label{CompTech2}Let $\mathcal{R}\subseteq \mathcal{C}$ such that $\mathcal{R}$ is closed with respect to $\tau(d^{s}_{\mathcal{C}})$. Let $\Phi:\mathcal{R} \rightarrow \mathcal{R}$ be a monotone mapping. If there exist $f,g\in \mathcal{R}$ such that the following assertions hold:
\begin{enumerate}
\item $g\preceq_{d_{\mathcal{C}}} \Phi(g)$ and $\Phi$ is orbitally $\preceq_{d_{\mathcal{C}}}$-continuous at $g$,
\item $g\preceq_{d_{\mathcal{C}}} f$ and  $\Phi(f)\preceq_{d_{\mathcal{C}}} f$.
\end{enumerate} Then there exists $f^{\star}\in \mathcal{R}$ such that $f^{\star}\in Fix(\Phi)$ and $f^{\star}\in \Omega(g)\cap \mathcal{O}(f)$.
\end{corollary}

\begin{proof}If $\mathcal{R}$ is closed with respect to $\tau(d^{s}_{\mathcal{C}})$, then $(\mathcal{R},d_{\mathcal{C}})$ is left Smyth complete, since $(\mathcal{C},d_{\mathcal{C}})$ is left Smyth complete. Proposition \ref{completeness} ensures that  $(\mathcal{R},d_{\mathcal{C}})$ is $\preceq_{d_{\mathcal{C}}}$-complete and, thus, Proposition \ref{usefullemma} gives that $(\mathcal{R},\preceq_{d_{\mathcal{C}}})$ is chain complete. Theorem \ref{CompTech} yields the desired conclusions. \end{proof}

In the following we show that Theorem \ref{SchT} can be recovered from Theorem \ref{CompTech}. To this end, we need the next sequence of useful results. The proof of the below lemma was given in \cite{LRV2018}.

\begin{lemma}\label{upper}Let $(X,d)$ be a quasi-metric space. If  $x$ is an upper bound of a sequence $(x_{n})_{n\in\mathbb{N}}$ in $(X,\preceq_{d})$ and, in addition, $(x_{n})_{n\in\mathbb{N}}$ converges to $x$  with respect to $\tau(d)$, then $x$ is the supremum of $(x_{n})_{n\in\mathbb{N}}$ in $(X,\preceq_{d})$.
\end{lemma}

Taking into account the above result we have the next one.

\begin{proposition}\label{Cont}Let $(X,d)$ be a $\preceq_{d}$-complete quasi-metric space and let $f:X\rightarrow X$ be a monotone mapping. Assume that there exists $x_{0}\in X$ such that $(f^{n}(x_{0}))_{n\in\mathbb{N}^\star}$ is increasing in $(X,\preceq_{d})$ and that $f$ is continuous from $(X,\tau(d))$ into itself, then $f$ is orbitally $\preceq_{d}$-continuous at $x_{0}$.
\end{proposition}

\begin{proof}Let $x_0\in X$ such that the sequence $(f^{n}(x_{0}))_{n\in\mathbb{N}^\star}$ is increasing in $(X,\preceq_{d})$.  Since the quasi-metric space $(X,d)$ is $\preceq_{d}$-complete there exists $x\in X$ such that the sequence $(f^{n}(x_{0}))_{n\in\mathbb{N}}$ converges to $x$ with respect to $\tau(d^{s})$. By Lemma \ref{least}, $x$ is the supremum of $(f^{n}(x_{0}))_{n\in\mathbb{N}}$. Moreover, the continuity of $f$ gives that $(f^{n+1}(x_{0}))_{n\in\mathbb{N}^\star}$ converges to $f(x)$ with respect to $\tau(d)$ and the monotony of $f$ provides that $f(x)$ is an upper bound of $(f^{n}(x_{0}))_{n\in\mathbb{N}^\star}$ in $(X,\preceq_{d})$. By Lemma \ref{upper} we have that $f(x)$ is the supremum of $(f^{n+1}(x_0))_{n\in\mathbb{N}^\star}$. Therefore $f$ is orbitally $\preceq_{d}$-continuous at $x_{0}$. \end{proof}

From the preceding result we can derive the following one which was proved in \cite{LRV2018}.

\begin{corollary}\label{Cont3}Let $(X,d)$ be a $\preceq_{d}$-complete quasi-metric space and let $f:X\rightarrow X$ be a mapping. If $f$ is continuous from $(X,\tau(d))$ into itself, then $f$ is $\preceq_{d}$-continuous.
\end{corollary}

In addition to the preceding results we have the next one which will be crucial in our subsequent discussion.

\begin{proposition}\label{Cont2}Let $(X,d)$ be a quasi-metric space and let $f:X\rightarrow X$ be a mapping. Assume that there exists $c\in [0,1[$ such that $$d(f(x),f(y))\leq cd(x,y)$$ for all $x,y\in X$. Then the following assertions hold:
\begin{enumerate}
\item $f$ is monotone $(X,\preceq_d)$ and continuous from $(X,\tau(d))$ into itself.

\item If there exist $v,w\in X$ with $v\preceq_{d}f(v)$ and $f(w)\preceq_{d} w$, then $v\preceq w$. 
\end{enumerate}
\end{proposition}
\begin{proof}$1.$ Let $x,y\in X$ with $x\preceq_{d} y$. Then $d(x,y)=0$. Since $d(f(x),f(y))\leq cd(x,y)$ we deduce that $d(f(x),f(y))=0$. Thus $f(x)\preceq_{d}f(y)$ and $f$ is monotone. Consider $x\in X$ and a sequence $(x_{n})_{n\in\mathbb{N}^\star}$ which converges to $x$ with respect to $\tau(d)$. Then $(f(x_{n}))_{n\in\mathbb{N}^\star}$ converges to $f(x)$ with respect to $\tau(d)$, since $d(f(x),f(x_{n}))\leq cd(x,x_{n})$ for all $n\in\mathbb{N}$. It follows that $f$ is continuous from $(X,\tau(d))$ into itself.

$2.$ Suppose that there exist $v,w\in X$ with $v\preceq_{d}f(v)$ and $f(w)\preceq_{d} w$. Then $d(v,f(v))=d(f(w),w)=0$. Hence we have that $$d(v,w)\leq d(v,f(v))+d(f(v),f(w))+d(f(w),w)\leq cd(v,w).$$ It follows that $d(v,w)=0$ and, thus, that $v\preceq_{d} w$, because otherwise we deduce that $1\leq c$ which is a contradiction. \end{proof}

By virtue of what is set out in the previous results, we are able to show that Theorems \ref{SchT} and \ref{SchT2} comes from Theorem \ref{CompTech} as it was announced. Indeed, the sets $\mathcal{C}_{b,c}$ and $\mathcal{C}_{c_{1}\ldots, c_{k}}$ were showed to be closed subsets of $\mathcal{C}$ with respect to $\tau(d^{s}_{\mathcal{C}})$ in \cite{RPV} and \cite{RomVal2013}, respectively. So the quasi-metric spaces $(\mathcal{C}_{b,c},d_{\mathcal{C}})$ and $(\mathcal{C}_{c_{1}\ldots, c_{k}},d_{\mathcal{C}})$ are left Smyth complete and, hence, $\preceq_{d_{\mathcal{C}}}$-complete. By Proposition \ref{Cont2} we have that the mappings $\Phi_{T}$, associated to (\ref{genreclinear}) and to (\ref{genreclinear2}), are monotone and continuous, since they are contractive, i.e., they satisfy that $$d_{\mathcal{C}}(\Phi_{T}(f),\Phi_{T}(g))\leq \frac{1}{2} d_{\mathcal{C}}(f,g)$$ for all $f,g\in \mathcal{C}_{b,c}$ and $$d_{\mathcal{C}}(\Phi_{T}(f),\Phi_{T}(g))\leq \left(\max_{1\leq i\leq k}\frac{1}{a_{i}}\right) \left( \frac{2^{k}-1}{2^{k}}\right) d_{\mathcal{C}}(f,g)$$ for all $f,g\in \mathcal{C}_{c_{1},\ldots,c_{k}}$.

Now, if there exists $g\in \mathcal{C}_{b,c}$ ($g\in \mathcal{C}_{c_{1}\ldots, c_{k}}$) such that $g\preceq_{d_{\mathcal{C}}} \Phi_{T}(g)$, then, by Proposition \ref{Cont},  $\Phi_{T}(g)$ is orbitally $\preceq_{d_{\mathcal{C}}}$-continuous at $g$. Moreover, if there exists $f\in \mathcal{C}_{b,c}$ ($f\in \mathcal{C}_{c_{1}\ldots, c_{k}}$) such that $\Phi_{T}(f)\preceq_{d_{\mathcal{C}}} f$ then Proposition \ref{Cont2} guarantees that $g\preceq_{d_{\mathcal{C}}} f$. Therefore Theorem \ref{CompTech} (or Corollary \ref{CompTech2}) provides that there exists $f^{\star}\in \mathcal{C}_{b,c}$ ($f^{\star}\in \mathcal{C}_{c_{1}\ldots, c_{k}}$) such that $f^{\star}\in \Omega(g)\cap \mathcal{O}(f)$.

Next we show that Theorem \ref{R} can be derived form Theorem \ref{CompTech} as promised. First, according to \cite{RomVal2013}, the quasi-metric space $(\mathcal{C}_{c_{1},\ldots,c_{k-1}}, d_{\mathcal{C}})$ is left Smyth complete and, hence, $\preceq_{d_{\mathcal{C}}}$-complete. So, by Proposition \ref{usefullemma}, we have that the partially ordered set $(X,\preceq_d)$ is chain complete.

It is clear that the mapping $\Phi_{T}$, given by (\ref{RecProb}), is monotone with respect to $\preceq_{d_{\mathcal{C}}}$. Moreover, $g_{h}\preceq_{d_{\mathcal{C}}}\Phi_{T}(g_{h})$, where $g_{h}\in \mathcal{C}_{c_{1},\ldots,c_{k}}$ with $g_{h}(n)=h(n)$ for all $n\geq k$ and $g_{h}(n)=c_{n}$ for all $1\leq n< k$. In fact, note that $g_{h}\preceq_{d_{\mathcal{C}}} \Phi(f)$ for all $f\in \mathcal{C}_{c_{1},\ldots,c_{k-1}}$.

Furthermore, $\Phi_{T}$ is orbitally $\preceq_{d_{\mathcal{C}}}$-continuous at $g_{h}$. Indeed we have that the sequence $(\Phi^{m}_{T}(g_{h}))_{m\in\mathbb{N}^\star}$ is increasing in $(\mathcal{C}_{c_{1},\ldots,c_{k-1}},\preceq_{d_{\mathcal{C}}})$ and $(\mathcal{C}_{c_{1},\ldots,c_{k}},\preceq_{d_{\mathcal{C}}})$ is chain complete and, thus, that there exists $f^{\star}\in \mathcal{C}_{c_{1},\ldots,c_{k-1}}$ such that $f^{\star}$ is the supremum of $(\Phi^{m}_{T}(g_{h}))_{m\in\mathbb{N}^\star}$ in $(\mathcal{C}_{c_{1},\ldots,c_{k-1}},\preceq_{d_{\mathcal{C}}})$. On the one hand, since $\Phi_T$ is monotone we have that $\Phi_T(f^\star)$ is an upper bound of the sequence $(\Phi^{m+1}_{T}(g_{h}))_{m\in\mathbb{N}^\star}.$ On the other hand, fixed $n\in\mathbb{N}$ such that $n>k$ we have that, for every $\varepsilon$, there exists $m_{\varepsilon}$ such that

$$f^{\star}(i)<\varepsilon+\Phi^{m_{\varepsilon}}_{T}(g_{h})(i)$$ for all $k\leq i\leq n-1$. Thus we obtain that 

$$ \begin{array}{ll}
\Phi_{T}(f^{\star})(n)<\sum_{i=k}^{n-1}v_{i}(n)\varepsilon+h(n)+\Phi^{m_{\varepsilon}}_{T}(g_{h})(n) = &\\
\\
\varepsilon \sum_{i=1}^{k-1}v_{i}(n)+\Phi^{m_{\varepsilon}+1}_{T}(g_{h})(n)\leq K\varepsilon+f^{\star}(n).&\\
\end{array}
$$ It follows that $\Phi_{T}(f^{\star})\preceq_{d_{\mathcal{C}}} f^{\star}$ and so $\Phi_{T}$ is orbitally $\preceq_{d_{\mathcal{C}}}$-continuous at $g_{h}$.

Now, if there exists $f\in \mathcal{C}_{c_{1},\ldots,c_{k-1}}$ such that $\Phi_T(f)\preceq_{d_{\mathcal{C}}} f$, then $g_{h}\preceq_{d_{\mathcal{C}}} \Phi_T(f)\preceq_{d_{\mathcal{C}}} f$. Whence we obtain, by Theorem \ref{CompTech} (or Corollary \ref{CompTech2}), that $f^{\star}\in \Omega(g_{h})\cap \mathcal{O}(f)$.

It is worthy to observe that Proposition \ref{useful}, the main result in which Theorem \ref{R} is based on, can be derived from  Lemma \ref{least} and Propositions \ref{completeness} and \ref{usefullemma}.

We end the paper, noting that Theorem \ref{CompTech} (and Corollary \ref{CompTech2}) introduces a fixed point technique for asymptotic complexity analysis of algorithms which does not assume requirements over all elements in a subset $\mathcal{R}$ of $\mathcal{C}$. It follows that we can reduce the set of elements over which we need to check those conditions that allow discuss the asymptotic complexity of an algorithm whose running time satisfies a recurrence equation. Hence the new technique improves those given in \cite{RPV,RomVal2013,Schellekens}. Besides, the aforementioned technique captures the essence of that given in Theorem \ref{R} and, in addition, it allows to state upper and lower asymptotic bounds for the running time computing of algorithms. So, in this sense, it improves the technique introduced in Theorem \ref{R}. Besides, the new fixed point method preserve the original Scott's ideas providing a common framework for Denotational Semantics and Asymptotic Complexity of algorithms.
 

\section{Future Work}
 { It must be stressed that in Scott's approach the $\preceq$-continuity of a mapping matches up with the notion of continuity with respect to the so-called Scott topology (see, for instance, \cite{JGL2013}). Clearly in our new context the $\preceq$-continuity has been replaced by the orbital $\preceq$-continuity. So it seems natural to wonder whether such a notion can be interpreted as a kind of continuity with respect to any topology. Thus the authors propose as future work to analyze if that topology exists and, if so, characterize it. }
 
 J.J. Mi\~nana and O. Valero acknowledge financial support from UE funds and Programa Operatiu FEDER 2014-2020 de les Illes Balears, by project PROCOE/4/2017 (Direcci\'{o} General d'Innovaci\'{o} i Recerca, Govern de les Illes Balears) and by project ROBINS. The latter has received research funding from the EU H2020 framework under GA 779776. This publication reflects only the authors views and the European Union is not liable for any use that may be made of the information contained therein. A. Estevan acknowledges financial support from the Ministry of Economy and Competitiveness of Spain under grants  MTM2015-63608-P (MINECO/FEDER) and ECO2015-65031.


\begin{thebibliography}{99}


\bibitem{Baranga}A. Baranga, \textit{The contraction principle as a particular case of Kleene's fixed point theorem}, Discrete Math. 98 (1991), 75-79.

\bibitem{BB1988} G. Brassard, P. Bratley, \textit{Algorithms: Theory and Practice}, Prentice-Hall, Englewood Cliffs, 1988.


\bibitem{Cobzas} \c{S}. Cobza\c{s}, \textit{Functional Analysis in Asymmetric Normed Spaces}, Birkh\"auser Basel, Dordrecht, 2013.


\bibitem{EsikRon}Z. \'Esik, P. Rondogiannis, \textit{A fixed point theorem for non-monotonic functions}, Theoret. Comput. Sci. 574 (2015), 18-38.


\bibitem{Cull} P. Cull, M. Flahive, R. Robson, \textit{Difference Equations: From Rabbits to Chaos}, Springer, New York, 2005.

\bibitem{FomenkoPodop16} T.N. Fomenko, D.A. Podoprikhin, \textit{Fixed points and coincidences of mapping of partially ordered sets}, J. Fixed Point Theory Appl.  18 (2016), 823-842.

\bibitem{FomenkoPodop17} T.N. Fomenko, D.A. Podoprikhin, \textit{Commom fixed points and coincidences of mapping families on partially ordered sets}, Topol. Appl. 221 (2017), 275-285.


\bibitem{Davey} B.A. Davey, H.A. Priestley, \textit{Introduction to Lattices
and Order, }Cambridge University Press, Cambridge, 1990.


\bibitem{GRRS2008}L.M. Garc\'{i}a-Raffi, S. Romaguera, and M.P. Schellekens, \textit{Applications of the complexity space to the general
probabilistic divide and conquer algorithms}, J. Math. Anal. Appl. 348 (2008), pp. 346-355.

\bibitem{JGL2013}J. Goubault-Larrecq, \textit{Non-Hausdorff Topology and Domain Theory}, Cambridge University Press, New York, 2013.


\bibitem{Scott2}C.A. Gunter, D.S. Scott, \textit{Semantic domains},  in: Handbook of Theoretical Computer Science, ed. by  J. van Leewen, Vol. B, MIT Press, Cambridge, 1990, pp. 663-674.

\bibitem{SedaHit2}P. Hitzler, A.K. Seda, Mathematical Aspects of Logic Programming Semantics,  CRC Press, Boca Raton, 2011.

\bibitem{Knuth1973} D.E. Knuth, \textit{The Art of Computer Programming}, Vol 3. Sorting and Searching, Addison-Wesley, Redwood, 1973.

\bibitem{Ku} H.P.A. K\"{u}nzi, \textit{Nonsymmetric distances and their
associated topologies: about the origins of basic ideas in the area of
asymmetric topology, }in: Handbook of the History of General Topology, ed.
by C.E. Aull and R. Lowen, Vol. 3, Kluwer, Dordrecht, 2001, pp. 853-968.



\bibitem{KSch} H.-P.A. K\"{u}nzi, M.P. Schellekens, \textit{On the Yoneda completion of a quasi-metric space}, Theoret. Comput. Sci. 278 (2002), 159-194.

\bibitem{LRV2018} M. L\'opez-Ram\'irez, O. Valero, \textit{Qualitative versus quantitative fixed point techniques in Computer Science},  Quaest. Math. \textbf{41} (2018), 115-127.

\bibitem{Manna}Z. Manna, \textit{Mathematical Theory of Computation}, McGraw-Hill, New York, 1974.

\bibitem{Ma} S.G. Matthews, \textit{Partial metric topology,} Ann. New York Acad. Sci. 728 (1994), 183-197.

\bibitem{RusPetrusel2008}A. Petrusel, G. Petrusel, I.A. Rus, \textit{Fixed Point Theory}, Cluj University Press, Cluj-Napoca, 2008.

\bibitem{Reilly}I.L. Reilly, P.V. Subrahmanyam, M.K. Vamanamrthy, \textit{Cauchy sequences in quasi-pseudo-metric spaces}, Mh. Math. 93 (1982), 127-140.


\bibitem{RPV}S. Romaguera, P. Tirado, O. Valero, \textit{New results on mathematical foundations of asymptotic complexity analysis of algorithms via complexity spaces}, Int. J. Comput. Math. 89 (2012), 1728-1741.


\bibitem{RomVal2013}S. Romaguera, O. Valero, \textit{A common mathematical framework for asymptotic complexity analysis and denotational semantics for recursive programs based on complexity spaces}. In:  M.T. Afzal (eds.), Semantics - Advances in Theories and Mathematical Models Vol. 1, pp. 99-120, InTech Open Science, Rijeka, (2012).


\bibitem{Schellekens} M.P. Schellekens, \textit{The Smyth completion: a common
foundation for the denotational semantics and complexity analysis, }  Electron. Notes Theor. Comput. Sci. 1 (1995), 211-232.

\bibitem{Scott} D. S. Scott, \textit{Outline of a mathematical theory of computation, }in: Proc. 4th Annual Princeton Conference on Information Sciences and Systems, 1970, 169-176.


\bibitem{Stoy}J.E. Stoy, \textit{Denotational Semantics: The Scott-Strachey Approach to Programming Language Theory}, MIT Press, Cambridge, 1977. 




\end{thebibliography}
\end{document}